\documentclass[11pt]{article}

\usepackage[margin=1in]{geometry}

\usepackage{amsmath,amssymb,amsfonts,amsthm}
\usepackage{algorithm}
\usepackage{algpseudocode}
\numberwithin{equation}{section}

\usepackage[numbers]{natbib} 

\usepackage{graphicx}
\graphicspath{{Fig/}}

\usepackage{hyperref}
\usepackage[nameinlink,noabbrev]{cleveref}

\theoremstyle{plain}
\newtheorem{theorem}{Theorem}

\newtheorem{lemma}[theorem]{Lemma}
\newtheorem{corollary}[theorem]{Corollary}
\newtheorem{remark}{Remark}%

\theoremstyle{definition}
\newtheorem{definition}{Definition}

\theoremstyle{remark}

\newcommand{\E}{\mathbb{E}}
\newcommand{\LL}{\mathcal{L}}

\newcommand{\X}{\mathcal{X}}

\newcommand{\bl}{\bar{\sigma}}
\newcommand{\dd}{\mathrm{d}}
\newcommand{\T}{\mathcal{T}}
\DeclareMathOperator*{\argmax}{arg\,max}
\DeclareMathOperator*{\argmin}{arg\,min}

\title{Hierarchical Maximum Entropy via the Renormalization Group}
\author{Amir R.~Asadi\\
\small Statistical Laboratory, Centre for Mathematical Sciences, University of Cambridge\\
\small Cambridge CB3 0WA, United Kingdom\\
\small \texttt{asadi@statslab.cam.ac.uk}}
\date{} 

\begin{document}
\maketitle

\abstract{Hierarchical structures, which include multiple levels, are prevalent in statistical and machine-learning models as well as physical systems. Extending the foundational result that the maximum entropy distribution under mean constraints is given by the exponential Gibbs-Boltzmann form, we introduce the framework of \emph{hierarchical maximum entropy} to address these multilevel models. We demonstrate that Pareto optimal distributions, which maximize entropies across all levels of hierarchical transformations, can be obtained via renormalization-group procedures from theoretical physics. This is achieved by formulating multilevel extensions of the Gibbs variational principle and the Donsker-Varadhan variational representation of entropy. Moreover, we explore settings with hierarchical invariances that significantly simplify the renormalization-group procedures, enhancing computational efficiency: quadratic modular loss functions, logarithmic loss functions, and nearest-neighbor loss functions. This is accomplished through the introduction of the concept of parameter flows, which serves as an analog to renormalization flows in renormalization group theory. This work connects ideas from probability theory, information theory, and statistical mechanics. 
}


\maketitle

\section{Introduction}\label{Introduction section} 
\subsection{Background}
The Maximum Entropy (ME) principle, rooted in information theory and statistical mechanics \cite{jaynes1957information, jaynes1957information2, csiszar1975divergence}, provides a robust framework for modeling probability distributions when only partial information is available. The core idea is to select the distribution that satisfies the known constraints, such as specified averages under a loss function, while making the fewest additional assumptions, by maximizing information entropy. Initially introduced by Claude Shannon and intimately related to thermodynamic entropy, information entropy quantifies the uncertainty inherent in a distribution. By maximizing entropy under the imposed constraints, the ME principle avoids injecting extraneous biases and yields the most uninformative distribution consistent with the data.

The ME principle and its associated technical tools have been effectively applied across various disciplines. In statistics and data science, it is used in statistical inference from data samples, as discussed by \cite{grunwald2004game} and \cite{phillips2004maximum}. Additionally, ME techniques are utilized in entropic optimal transport and Schr\"{o}dinger bridges, which are foundational for score-based generative modeling, as explored by \cite{nutz2022entropic}, and \cite{de2021diffusion}. In natural language processing, ME models are employed in tasks such as text classification, part-of-speech tagging, and language modeling. Notable works include \cite{berger1996maximum}, \cite{ratnaparkhi-1996-maximum}, and \cite{nigam1999using}. In probabilistic machine learning, the ME principle underpins entropic regularization techniques, providing principled methods for deriving optimal posterior distributions during model training, as detailed by \cite{Catoni2007} and \cite{alquier2024user}. In physics and statistical mechanics, it is fundamental to the formulation of equilibrium models, a concept introduced by \cite{jaynes1957information}. Furthermore, the ME principle is closely related to the Principle of Minimum Information Discrimination and the Principle of Maximum Cross Entropy, which are essential in deriving probability distributions that align with given constraints while remaining as unbiased as possible \cite{soofi2000principal}.

A fundamental result in probability and information theory is the emergence of the Gibbs-Boltzmann distribution as the solution to the entropy maximization problem under moment constraints. Specifically, given a random variable \( X \) and a loss function \( L \) (also referred to as the energy function in statistical physics), assume that we seek to determine the probability distribution $P_X$ that maximizes entropy while satisfying a constraint on the expected value of \( L(X) \). This can be formulated as the optimization problem:
\begin{align}\label{eq: max_entropy}
    \argmax_{P_X \,:\, \E[L(X)] = \mu} H(X).
\end{align}
where $H(X)$ denotes the entropy of $P_X$ (see Section \ref{sec: prelim}).
The optimal solution is given by the Gibbs-Boltzmann distribution
\[
\tilde{P}_X(x) \propto \exp\left(\lambda L(x)\right),
\]
where the temperature parameter \( \lambda \) is chosen to satisfy the constraint \( \E[L(X)]=\mu \).
This result follows from the Gibbs variational principle or the Donsker-Varadhan variational representation of entropy (see Lemma \ref{Gibbs is the maximizer entropy}), both of which provide an optimization-based characterization of equilibrium distributions. In this paper, we extend this framework to hierarchical settings, incorporating entropy at multiple levels of a hierarchy and generalizing the classical Gibbs-Boltzmann formulation. 

In many real-world systems, structures exhibit dependencies and interactions across multiple length scales, leading to complex hierarchical organizations. Such multi-scale systems and hierarchical models are ubiquitous in physics, machine learning, and statistics. However, the classical ME framework typically does not account for uncertainty across multiple hierarchies or scales.  
Here, we introduce \emph{hierarchical maximum entropy}, a framework that generalizes the classical ME framework by integrating hierarchical structures into the entropy maximization problem. This enables us to study distributions that maximize uncertainty over multiple levels of a hierarchy while still satisfying mean constraints. We then show close intrinsic connections with renormalization group theory.

In physics, problems with multiple scales, such as critical phenomena, are usually analyzed with renormalization group theory \cite{wilson1983renormalization}.
The renormalization group, formulated by Wilson in the 1970s, is a powerful apparatus for studying complex problems involving multiple length scales, particularly critical phenomena in statistical physics \cite{wilson1974renormalization}. Beyond physics, it has been linked to probability theory, where it helps explain universality \cite[Chapter 10]{koralov2007theory}\cite{jona2001renormalization}. The renormalization group provides a systematic approach to analyzing multiscale problems by exploiting the fundamental concept of {self-similarity} (or scale invariance). Many challenging physics problems arise due to interactions across multiple scales, making direct analysis difficult. The renormalization group tackles this by decomposing a complex, multiscale system into a sequence of simpler steps, each associated with a single characteristic length scale, thus making computations more tractable. Implementation in computers was one of the main motivations of Wilson in developing renormalization group theory \cite{wilson1983renormalization}. At the heart of this framework is iterative coarse-graining and renormalization of a system, and using the idea of \emph{renormalization flow}, which describes how model parameters evolve under these iterative transformations. This flow enables systematic exploration of the behavior of the system on different scales, effectively treating scale as an analog of time \cite{weinan2011principles}. Using this structured approach, the renormalization group provides deep insights into the fundamental nature of criticality, universality, and other emergent phenomena in physics and mathematics.

In this paper, we define hierarchical entropy as the linear combination of entropies computed at different levels of a hierarchical system, given a vector of weight coefficients. We study the optimization of hierarchical entropy as an essential step toward our ultimate goal of analyzing the framework of hierarchical maximum entropy.  
There are conceptual similarities between the definition of hierarchical entropy and what the works by \cite{zhang1991complexity} and \cite{fogedby1992phase} refer to as ``phase space complexity'' in statistical mechanics. The characterization of maximum phase space complexity distributions in general settings was left as an open question in \cite{fogedby1992phase}, where it was suggested that a renormalization group approach could be considered to evaluate complexity as a function of temperature.

\subsection{Contributions of the Paper}
In this paper, we investigate the properties and applications of hierarchical maximum entropy. Our contributions are specifically as follows:
\begin{enumerate}
    \item We formally define the hierarchical maximum entropy problem and develop a rigorous theoretical framework for its analysis. We also study the related problem of hierarchical minimum relative entropy.
    \item We demonstrate that solutions to this problem can be constructed via the renormalization group procedure in conjunction with disintegration operations.
    \item We identify several settings where hierarchical invariances simplify the iterative renormalization group procedure by introducing the concept of parameter flows, leading to significant computational advantages in computing the optimal distributions. In particular, we analyze cases involving quadratic modular loss functions, logarithmic loss functions, and nearest-neighbor loss functions.
\end{enumerate}

By bridging ideas from probability theory, information theory, and statistical mechanics, our work offers new insights into the entropy-maximizing distributions in hierarchical systems.

\subsection{Organization of the Paper}
The remainder of the paper is organized as follows. Section \ref{sec: prelim} offers several notations, preliminary definitions, and tools. In Section \ref{sec:problem_description}, we introduce the hierarchical maximum entropy problem, providing formal definitions and a detailed problem setup. Section \ref{sec:renormalizations} presents the iterative renormalization group procedure used to solve the hierarchical maximum entropy problem along with related results. In Section \ref{sec:hierarchical_invariances}, we present examples where hierarchical invariances arise: (i) quadratic modular loss functions, using modular multivariate Gaussians, (ii) logarithmic loss functions, showcasing Dirichlet distributions, and (iii) nearest-neighbor loss functions, utilizing the self-similarity of the one-dimensional Ising model. For each example, we analyze the parameter flows and the behavior of the iterative procedures. Section \ref{sec: relative entropy} discusses the related hierarchical relative entropy minimization problem. Finally, Section \ref{sec:conclusions} concludes the paper, summarizing the main findings and discussing potential future research directions.  
\section{Preliminaries}\label{sec: prelim}
Given any measurable mapping $f$ and any measure $\mu$, we denote the pushforward measure of $\mu$ by $f$ with $f_{\#}\mu$. We write $X\sim P$ to denote that random variable $X$ has distribution $P$. The notation $P_X\ll Q_X$ signifies that the distribution $P_X$ is absolutely continuous with respect to the distribution $Q_X$. We write $A\succ 0$ if $A$ is a positive definite matrix. The notation $\text{Dir}(\alpha_1, \alpha_2, \dots, \alpha_N)$ refers to a Dirichlet distribution with parameters $\alpha_1, \alpha_2, \dots, \alpha_N.$
\begin{definition}  
Let \(\mathcal{A}\) be a set, \(P_X\) be a probability distribution defined on \(\mathcal{A}\), and $X$ a random variable (or vector) with distribution $P_X$.   
    If \(P_X\) is a discrete probability distribution, the (Shannon) entropy of $X$ is defined as  
    \[
        H(X)=H(P_X) = -\sum_{x \in \mathcal{A}} P_X(x) \log P_X(x),
    \] 
    where \(P_X(x) \log P_X(x)\) is understood to be \(0\) whenever \(P_X(x) = 0\).    
    If \(P_X\) is a continuous probability distribution with probability density function \(p_X(x)\), the (differential) entropy of $X$ is defined as  
    \[
        H(X)=H(P_X) = -\int_{\mathcal{A}} p_X(x) \log p_X(x) \, \mathrm{d}x,
    \]  
    where \(p_X(x) \log p_X(x)\) is understood to be \(0\) wherever \(p_X(x) = 0\).  
\end{definition}
In this paper, all logarithms are in the natural base. Therefore, all entropies are in nat units.
\begin{definition}  
    The relative entropy (or Kullback-Leibler divergence) between two probability distributions \(P_X\) and \(Q_X\) defined on the same set \(\mathcal{A}\) is defined as
    \begin{equation*}
        D(P_X\|Q_X)=\E\left[\frac{\dd P_X}{\dd Q_X}(X) \right], 
    \end{equation*}
    where $X\sim P_X$ and $\frac{\dd P_X}{\dd Q_X}$ denotes the Radon-Nikodym derivative if $P_X\ll Q_X$; otherwise, $D(P_X\|Q_X)=\infty$.
\end{definition}
\begin{definition}
    The conditional relative entropy between $P_{Y|X}$ and $Q_{Y|X}$ given $P_X$ is defined as
    \[D(P_{Y|X}\|Q_{Y|X}|P_X):=\E \left[D(P_{Y|X}(\cdot|X)\|Q_{Y|X}(\cdot|X))\right], \qquad X\sim P_X.\]
\end{definition}
The following result, relating the relative entropy and conditional relative entropy, is known as the chain rule.
\begin{lemma}[Chain rule]\label{lem: chain rule entropy}
Let $X$ and $Y$ be two random vectors. Then,
\[H(X,Y)=H(X)+H(Y|X).\]
    Moreover, if $P_{XY}$ and $Q_{XY}$ are two joint distributions of $X$ and $Y$, then
    \[D(P_{XY}\|Q_{XY})=D(P_X\|Q_X)+D(P_{Y|X}\|Q_{Y|X}|P_X).\]
\end{lemma}
We now state the following well-known result, closely related to the Donsker-Varadhan variational representation of entropy, which we will later extend to hierarchical settings:
\begin{lemma}[Gibbs variational principle]\label{Gibbs is the maximizer entropy}
  Let $f: {\mathcal A}\rightarrow \mathbb{R}$ be such that \[\tilde{Z}:=\int_{x\in {\mathcal A}}\exp(-\lambda f(x))\,\dd x<\infty,\] where ${\mathcal A}$ is a continuous set. Then for any $P_{X}$ defined on ${\mathcal A}$,
  \begin{equation}
      \nonumber
      H(X)-\lambda\mathbb{E}[f(X)]= \log \tilde{Z}-D\left(P_{X}\middle\|\tilde{P}_{X}\right),
  \end{equation}
  where $X\sim P_X$, and $\tilde{P}_X$ is the distribution with probability density function
  \begin{equation*}
      \tilde{p}_{X}(x):=\frac{\exp(-\lambda f(x))}{Z},\textrm{ for all } x\in \mathcal{A}.
  \end{equation*}
\end{lemma}
An analogous counterpart to Lemma \ref{Gibbs is the maximizer entropy} applies to discrete random variables by replacing integrals with sums and interpreting $H(X)$ as Shannon entropy.

In multi-objective (vector) optimization problems, we consider several conflicting objectives simultaneously. Unlike single-objective optimization, where the goal is to find the best solution, multi-objective optimization seeks solutions that balance trade-offs among different criteria. In such settings, a solution is considered desirable if it is not possible to improve one objective without causing a decline in another. This idea leads to the concept of Pareto optimality, which provides a framework for evaluating and comparing solutions based on whether one solution `dominates' another. With this understanding, we can now present the precise definition of Pareto optimality and the Pareto front.
\begin{definition}[Pareto Optimality and Pareto Front]
Let $\Omega$ be a feasible set and let $\{f_i: \Omega \to \mathbb{R}\}_{i=1}^d$ be a collection of objective functions. A point $x^\ast \in \Omega$ is called \emph{Pareto optimal} and denoted with
\[
x^\ast=\argmax_{x\in \Omega} ~(f_1(x),\dots,f_d(x)),
\]
if there is no other point $x \in \Omega$ such that
\[
f_i(x) \ge f_i(x^\ast) \quad \text{for all } i = 1,\dots,d,
\]
with the strict inequality
\[
f_j(x) > f_j(x^\ast)
\]
holding for at least one index $j\in\{1,\dots,d\}$. In other words, no other feasible point can improve one objective without causing a deterioration in at least one other objective. The \emph{Pareto front} is the set of all Pareto optimal points, that is,
\[
\mathcal{P} = \{ x \in \Omega : x \text{ is Pareto optimal} \}.
\]
\end{definition}
This formulation emphasizes that a Pareto optimal point cannot be dominated by any other feasible point across all objectives, aligning with the general concept of efficiency in multi-objective optimization \cite{miettinen1999nonlinear}.
\section{Hierarchical Maximum Entropy}\label{sec:problem_description}
Assume that $X$ is a random variable or a random vector taking values from the set $\mathcal{X}$, which represents data or the state of a system.  Consider a hierarchical sequence of (coarse-graining) transformations $\mathbf{T}:=(T_i)_{i=1}^{d-1}$, where we recursively define
\begin{equation*}
	X^{(i+1)}:= T_i\bigl(X^{(i)}\bigr), \quad \textrm{for all } 1\leq i\leq d-1,
\end{equation*} 
with $X^{(1)}:=X$. For all $1\leq i \leq d$, let $P_{X^{(i)}}$ denote the distribution of $X^{(i)}$. 

For example, $X$ may represent the pixel values of an image, where each transformation $T_i$ computes the average of neighboring pixel groups and outputs a single pixel per group, producing a lower-resolution image. Consequently, $X^{(2)},\dots, X^{(d)}$ are progressively lower-resolution representations of $X$. Another example of a transformation sequence $\mathbf{T}$ is called \emph{decimation}, where a subset of random variables is removed at each stage, leaving the rest intact. For example, let $X=(X_1,\dots,X_d)$ represent a collection of $d$ blocks. In this context, $X$ could correspond to the synaptic weights of an artificial neural network, with each block representing an individual layer. The transformations $\{T_i\}_{i=1}^{d-1}$ progressively reduce $X$ by simply eliminating the blocks sequentially, starting from $X_d$. Thus, for $1 \leq i \leq d$, the reduced representation is given by 
\[X^{(i)} = (X_1, \dots, X_{d-i+1}), \quad \textrm{for } 1\leq i \leq d.\]
In the context of renormalization group theory, decimation is a scaling transformation introduced by  \cite{kadanoff1975numerical} and \cite{migdal1975recursion}. Beyond physics, decimation also appears in cartography, where it is employed to simplify geographic maps of a region. As one zooms out, smaller cities are omitted, producing maps at progressively coarser scales.

In the sequel, assuming that all probability spaces are standard, we address the following multi-objective (vector) optimization problem, which we call hierarchical maximum entropy:
\begin{equation}\label{eq: original vector opt problem}
    \argmax_{P_X:~ \E[L(X)]=\mu} \bigl(H\bigl(P_{X^{(1)}}\bigr),H\bigl(P_{X^{(2)}}\bigr),\dots ,H\bigl(P_{X^{(d)}}\bigr)\bigr),
\end{equation}
where $H\bigl(P_{X^{(i)}}\bigr)$ represents the entropy of the $i$-th transformed random vector. The objective is to identify Pareto optimal distributions  $P^\ast_X$  that simultaneously maximize the entropies across all levels of the hierarchy.
Clearly, the classical maximum entropy problem \eqref{eq: max_entropy} is a special case of \eqref{eq: original vector opt problem} when $d=1$.
\begin{definition} For any given $X,\mathbf{T}$, and $\mathbf{\sigma}=(\sigma_1,\dots,\sigma_d)$, with $\sigma_i$ as positive real value for all $1\leq i\leq d$, we define \emph{hierarchical entropy} as 
\begin{equation*}
	H_{(\mathbf{\sigma},\mathbf{T})}(X):= \sum_{i=1}^d {\sigma_i}H\left(X^{(i)}\right).
\end{equation*} 
    We call $\mathbf{\sigma}$ the hierarchical coefficients vector, where each $\sigma_i$ represents the coefficient at level $i$. 
    
    For all $1\leq i \leq d$, let the {accumulated hierarchical coefficients} be defined as 
\begin{equation}\label{eq: accumulated coeff}
    \bar{\sigma}_i:=\sum_{k=1}^i \sigma_k.
\end{equation}
\end{definition}    
To solve \eqref{eq: original vector opt problem}, we employ the linear scalarization technique from multi-objective optimization literature (for a formal definition of linear scalarization, see \cite{hwang2012multiple}).
    Specifically, we aim to solve: 
\begin{align}\label{eq: linear scalarization problem}
    \argmax_{P_X:\E[L(X)]=\mu} H_{(\mathbf{\sigma},\mathbf{T})}(X),
\end{align}
where $\sigma_i>0$ are arbitrary scalar weights.  
\begin{theorem}
    The distribution $P_X$ is a Pareto optimal solution to \eqref{eq: original vector opt problem} if and only if there exist positive weights $\sigma_i$ such that $P_X$ is a solution to \eqref{eq: linear scalarization problem}.
\end{theorem}
\begin{proof}
    Since the entropy $H(X^{(i)})$ is concave in $P_{X^{(i)}}$, and $P_{X^{(i)}}$ is linear in $P_X$, we conclude that $H(X^{(i)})$ is concave in $P_{X}$. Therefore, since the objectives are all concave in $P_X$, and the set of all distributions $P_X$ with $\E[L(X)]=\mu$ is convex, linear scalarization constructs the entire Pareto front \cite{geoffrion1968proper}; see also \cite[Section 3.1]{miettinen1999nonlinear}.
\end{proof}

For the constrained optimization problem \eqref{eq: linear scalarization problem}, the Lagrangian is
\begin{align*}
	\LL(P_X)= H_{(\sigma,\mathbf{T})}(X)-\lambda \E[L(X)],
\end{align*}
where $\lambda$ is the Lagrange multiplier.
Next, we study the solution to the following associated problem, which we call the hierarchical entropy regularization problem:
\begin{equation}\label{eq: Lagrangian problem}
	\argmax_{P_X} \left\{H_{(\sigma,\mathbf{T})}(X)-\lambda \E[L(X)]\right\}.
\end{equation}

\section{Renormalization Group}\label{sec:renormalizations}
Consider the following definition of a probability operator.
\begin{definition}[Renormalization] Let $P_X$ be a continuous distribution on a set $\mathcal{A}$ with probability density function $p_X(x)$ and assume that $\theta>0$. Given inputs $P_X$ and $\theta$, we define the renormalization operator $(Q_X,Z)=\mathcal{R}(P_X;\theta)$ where
\begin{equation*}
    Z:=\int_{\mathcal{A}}(p_X(x))^{\theta}\dd x,
\end{equation*}
and $Q_X$ is the distribution with probability density function 
\begin{equation}
	\nonumber
	q_X(x):= \frac{(p_X(x))^{\theta}}{Z}, \textrm{ for all }x\in \mathcal{A}.
\end{equation}
The renormalization operator is defined analogously for discrete distributions except by replacing the integral with a sum.
\end{definition}
The distribution $Q_X$ defined above is also known as the escort distribution in the statistical physics literature; see, e.g., \cite{bercher2012simple}.

  Recall the definition of {accumulated hierarchical coefficients} in \eqref{eq: accumulated coeff} and
consider the iterative sequence of pushforwards (coarse-grainings) and renormalizations in Algorithm \ref{Hierarchical Entropy Algorithm}. Assume that all alphabets are standard Borel spaces, which guarantees the existence of regular conditional probabilities and reverse random transformations \cite{faden1985existence}.
    By disintegrating distributions $P^{(i)}_{X^{(i)}}$ for $i=1,\dots,d-1$, we define
\begin{equation}\label{eq: definition of P star}
    \tilde{P}^{[\lambda]}_X := P^{(d)}_{X^{(d)}}P^{(d-1)}_{X^{(d-1)}|X^{(d)}}\cdots P^{(1)}_{X^{(1)}|X^{(2)}}.
\end{equation}
Moreover, let
\begin{equation}\label{eq: definition of Z}
    Z(\lambda):=\prod_{i=1}^d Z_i.
\end{equation}
The first main result of this section, Theorem \ref{Max hierarchical Shannon entropy} below, shows that $\tilde{P}^{[\lambda]}_X$ is the solution to \eqref{eq: Lagrangian problem} and $\LL\bigl(\tilde{P}^{[\lambda]}_X\bigr)=\log Z(\lambda)$.
\begin{algorithm}[t]
    \caption{Renormalization Group}\label{Hierarchical Entropy Algorithm}
\begin{algorithmic}[1]
    \State \textbf{let} \(Z_1:=\int_{\X} \exp{\left(-\frac{\lambda }{\sigma_1}L(x)\right)}\mathrm{d}x \) and \(p^{(1)}_{X^{(1)}}( x):=\frac{\exp{\left(-\frac{\lambda }{\sigma_1}L(x)\right)}}{Z_1}\)
    \Comment{Initialization}
    \For{\(i = 1 \text{ to } d-1\)}
        \State \(U^{(i)}_{X^{(i+1)}} \gets (T_i)_{\#}P^{(i)}_{X^{(i)}}\) 
        \Comment{Pushforward (Coarse-graining)}
        \State \(\left(P^{(i+1)}_{X^{(i+1)}},Z_{i+1}\right) \gets \mathcal{R}\left(U^{(i)}_{X^{(i+1)}};\frac{\bar{\sigma}_{i}}{\bar{\sigma}_{i+1}}\right)\) 
        \Comment{Renormalization}
    \EndFor
\end{algorithmic}
\end{algorithm}
Before proving this theorem, we need the following lemmas. 

\begin{lemma}\label{Shannon entropy KL sum}
	Let $P_X$ and $Q_X$ be two continuous distributions taking values on the set $\mathcal{A}$, where $X\sim P_X$. Assume that $\theta \geq 0$ and  $(\tilde{Q}_X,\tilde{Z}):=\mathcal{R}\left(Q_X;\frac{\theta}{1+\theta}\right)$ is the output of the renormalization operator. We have 
	\begin{equation}\nonumber
	H(X)-\theta D(P_X\|Q_X) = \log \tilde{Z}-(1+\theta)D\left(P_X\middle\|\tilde{Q}_X\right).
	\end{equation}
\end{lemma}
\begin{proof}
Let $p_X(x)$ and $q_X(x)$ denote the probability density functions of $P_X$ and $Q_X$, respectively. We can write
  \begin{align*}
  H(X)-\theta D(P_X\|Q_X)
    &=-\int_{\mathcal{A}} p_X(x)\log p_X(x)\dd x  -\theta\int_{ \mathcal{A}} p_X(x)\log\left(\frac{p_X(x)}{q_X(x)}\right)\dd x\\
  &=-\int_{\mathcal{A}} p_X(x)\log\left(\frac{\left(p_X(x)\right)^{1+\theta}}{(q_X(x))^{\theta}}\right)\dd x \\
  &=-(1+\theta)\int_{\mathcal{A}} p_X(x)\log\left(\frac{p_X(x)}{(q_X(x))^{\frac{\theta}{1+\theta}}}\right)\dd x \\
  &=\log \tilde{Z}-(1+\theta)D\left(P_X\middle\|\tilde{Q}_X\right).
    \end{align*}
\end{proof}
An analogous result applies to discrete random variables, where $H(P)$ corresponds to Shannon entropy.

The following result can be viewed as an extension of Lemma \ref{lem: chain rule entropy}, representing the chain rule of relative entropy with respect to an arbitrary mapping. 

\begin{lemma}\label{chain rule based}
Let $P_{X_2}= T_{\#}P_{X_1}$ and $Q_{X_2}= T_{\#}Q_{X_1}$. Then
    \begin{equation}
        \nonumber
        D(P_{X_1}\|Q_{X_1}) = D(P_{X_2}\|Q_{X_2}) + D(P_{X_1|X_2}\|Q_{X_1|X_2}|P_{X_2}).
    \end{equation}
\end{lemma}
\begin{proof}
Expanding $D(P_{X_1X_2}\|Q_{X_1X_2})$ in two different ways based on the chain rule of relative entropy gives
\begin{align}
    D(P_{X_1X_2}\|Q_{X_1X_2}) &= D(P_{X_2}\|Q_{X_2}) + D(P_{X_1|X_2}\|Q_{X_1|X_2}|P_{X_2})\\
                              &= D(P_{X_1}\|Q_{X_1}) + D(P_{X_2|X_1}\|Q_{X_2|X_1}|P_{X_1}).
\end{align}
The conclusion follows from noting that $D(P_{X_2|X_1}\|Q_{X_2|X_1}|P_{X_1})=0$.
\end{proof}
\begin{definition}
    The hierarchical relative entropy between any two distributions $P_X$ and $Q_X$ is defined as
    \begin{equation}
    \nonumber
    	D_{(\sigma, \mathbf{T})}(P_X\|Q_{X}) := \sum_{i=1}^d \sigma_i D\left(P_{X^{(i)}}\middle\|Q_{{X}^{(i)}}\right),
    \end{equation}
    where $P_{X^{(i)}}$ and $Q_{X^{(i)}}$ are the pushforward measures when $P_X$ and $Q_X$ are passed through the sequence of transformations $\mathbf{T}=(T_i)_{i=1}^{d-1}$, respectively.
\end{definition}
The result of Lemma \ref{chain rule based} yields the following corollary that provides alternative representations of hierarchical entropies using conditional entropies and accumulated hierarchical coefficients:
\begin{corollary} 	
\label{lem: based on chain rule of entropy}
 We have 
 \begin{align*}
     H_{(\sigma, \mathbf{T})}(X) := \sum_{i=1}^{d-1} \bar{\sigma}_i H\left(X^{(i)}\middle|X^{(i+1)} \right) + \bl_d H\left(X^{(d)} \right),
 \end{align*}
 and
 \begin{align*}
     D_{(\sigma, \mathbf{T})}(P_X\|Q_{X}) := \sum_{i=1}^{d-1} \bar{\sigma}_i D\left(P_{X^{(i)}|X^{(i+1)}}\middle\|Q_{{X}^{(i)}|X^{(i+1)}}\middle|P_{X^{(i+1)}}\right)+\bar{\sigma}_d D\left(P_{X^{(d)}}\middle\|Q_{{X}^{(d)}}\right).
 \end{align*}
\end{corollary}

Now, we are ready to give the first main result of this section, which can be interpreted as a hierarchical extension of the Gibbs variational principle and the Donsker-Varadhan variational representations of entropy. 

\begin{theorem}\label{Max hierarchical Shannon entropy}
Consider Algorithm \ref{Hierarchical Entropy Algorithm} and equations \eqref{eq: definition of P star} and \eqref{eq: definition of Z}. For any $P_X$ defined on $\mathcal{X}$ where $X\sim P_X$, we have
\begin{equation*}
		H_{(\sigma,\mathbf{T})}(X)-\lambda \E[L(X)]=\log Z(\lambda)-D_{(\sigma,\mathbf{T})}\Bigl(P_X\Big\|\tilde{P}^{[\lambda]}_X\Bigr).
	\end{equation*}
\end{theorem}
\begin{proof}
Based on induction on $k$, we first prove that
\begin{align}
   &\sum_{i=1}^k \sigma_iH\bigl(X^{(i)}\bigr)-\lambda\mathbb{E}[L(X)]=\nonumber\\
    &\quad \sum_{i=1}^k \log Z_i - \left(\bar{\sigma}_{k}D\left(P_{X^{(k)}}\middle\|P^{(k)}_{X^{(k)}}\right)+\sum_{i=1}^{k-1} \bar{\sigma}_i D\left(P_{X^{(i-1)}|X^{(i)}}\middle\|P^{(i-1)}_{X^{(i-1)}|X^{(i)}}\middle|P_{X^{(i)}}\right)\right).\label{eq: induction on k}
\end{align}
For $k=1$, based on Lemma \ref{Gibbs is the maximizer entropy}, we have 
\begin{align*}
    \sigma_1 H(X^{(1)})-\lambda \mathbb{E}[L(X)]&=\log Z_1-
    \sigma_1 D\left(P_{X^{(1)}}\middle\|P^{(1)}_{X^{(1)}}\right)\\
    &=\log Z_1-
    \bar{\sigma}_1 D\left(P_{X^{(1)}}\middle\|P^{(1)}_{X^{(1)}}\right).
\end{align*}
Assume that \eqref{eq: induction on k} holds for an arbitrary $1\leq k\leq d-1$.
We deduce that
\begin{align}
   & \sum_{i=1}^{k+1} \sigma_iH\left(X^{(i)}\right)-\lambda\mathbb{E}[L(X)]\nonumber\\
    &=\sigma_{k+1}H\left(X^{(k+1)}\right)+\sum_{i=1}^k \log Z_i\nonumber\\
    &\qquad -\left(\bar{\sigma}_{k}D\left(P_{X^{(k)}}\middle\|P^{(k)}_{X^{(k)}}\right)+\sum_{i=1}^{k-1} \bar{\sigma}_i D\left(P_{X^{(i-1)}|X^{(i)}}\middle\|P^{(i-1)}_{X^{(i-1)}|X^{(i)}}\middle|P_{X^{(i)}}\right)\right)\nonumber\\
    &=\sigma_{k+1}H\left(X^{(k+1)}\right) + \sum_{i=1}^k \log Z_i\nonumber\\
    &\qquad -\left(\bar{\sigma}_{k}D\left(P_{X^{(k+1)}}\middle\|U^{(k)}_{X^{(k+1)}}\right)+\sum_{i=1}^{k} \bar{\sigma}_i D\left(P_{X^{(i-1)}|X^{(i)}}\middle\|P^{(i-1)}_{X^{(i-1)}|X^{(i)}}\middle|P_{X^{(i)}}\right)\right)\label{eq: first identity}\\
    &=\sum_{i=1}^k \log Z_i +\sigma_{k+1}\left(H\left(X^{(k+1)}\right) - \frac{\bar{\sigma}_{k}}{\sigma_{k+1}}D\left(P_{X^{(k+1)}}\middle\|U^{(k)}_{X^{(k+1)}}\right) \right)\nonumber\\
    &\qquad -\sum_{i=1}^k \bar{\sigma}_i D\left(P_{X^{(i-1)}|X^{(i)}}\middle\|P^{(i-1)}_{X^{(i-1)}|X^{(i)}}\middle|P_{X^{(i)}}\right)\nonumber\\
    &=\sum_{i=1}^{k+1} \log Z_i \nonumber\\
    &\qquad -\left(\bar{\sigma}_{k+1}D\left(P_{X^{(k+1)}}\middle\|P^{(k+1)}_{X^{(k+1)}}\right)+\sum_{i=1}^{k} \bar{\sigma}_i D\left(P_{X^{(i-1)}|X^{(i)}}\middle\|P^{(i-1)}_{X^{(i-1)}|X^{(i)}}\middle|P_{X^{(i)}}\right)\right)\label{eq: third identity},
\end{align}
where \eqref{eq: first identity} follows from Lemma \ref{chain rule based}, and \eqref{eq: third identity} follows from Lemma \ref{Shannon entropy KL sum},
which proves \eqref{eq: induction on k} for $k+1$ and completes the inductive argument.
Therefore, for $k=d$, we have 
\begin{align}
    &H_{(\sigma,\mathbf{T})}(X)-\lambda \E[L(X)]\nonumber \\&= \sum_{i=1}^d \log Z_i - \bar{\sigma}_{d}D\left(P_{X^{(d)}}\middle\|P^{(d)}_{X^{(d)}}\right)- \sum_{i=1}^{d-1} \bar{\sigma}_i D\left(P_{X^{(i+1)}|X^{(i)}}\middle\|P^{(i)}_{{X}^{(i+1)}|X^{(i)}}\middle|P_{X^{(i)}}\right)\nonumber\\
    &=\log Z- D_{(\sigma,\mathbf{T})}\left(P_X\middle\|\tilde{P}^{[\lambda]}_X\right),\label{eq: follows from chain rule}
\end{align}
where \eqref{eq: follows from chain rule} follows from Corollary \ref{lem: based on chain rule of entropy}.
\end{proof}
Due to the non-negativity of hierarchical relative entropy, we can immediately deduce the following result:
\begin{corollary}\label{cor: maximization solution}
    The distribution $\tilde{P}^{[\lambda]}_X$ is the unique solution to the maximization problem (\ref{eq: Lagrangian problem}) and 
\begin{equation*}
    \max_{P_X} \left\{H_{(\sigma,\mathbf{T})}(X)-\lambda \E[L(X)]\right\} = \log Z(\lambda).
\end{equation*}
\end{corollary}
 
Let $\lambda^\ast$ be the Lagrange multiplier for which $\tilde{P}_X^{[\lambda^\ast]}$ solves the optimization problem \eqref{eq: linear scalarization problem}. The following result gives a condition that characterizes $\lambda^\ast$, allowing it to be solved for explicitly:
\begin{theorem}\label{thm: the optimal lambda} Given $X\sim \tilde{P}_X^{[\lambda^\ast]}$, assume that $H_{(\sigma,\mathbf{T})}(X)$ and $\E[L(X)]$ are differentiable with respect to $\lambda$.
 Then $\lambda^\ast$ is determined by the condition
\[
\frac{\mathrm{d}}{\mathrm{d}\lambda}\log Z(\lambda)\Big|_{\lambda=\lambda^\ast} = -\mu.
\]
\end{theorem}

\begin{proof}
By Corollary~\ref{cor: maximization solution}, the optimality condition reads
\[
\frac{\mathrm{d}}{\mathrm{d}\lambda}\Bigl( H_{(\sigma,\mathbf{T})}\bigl(\tilde{P}_X^{[\lambda]}\bigr) - \lambda^\ast\,\E_{\tilde{P}_X^{[\lambda]}}[L(X)] \Bigr)\Bigg|_{\lambda=\lambda^\ast} = 0.
\]
Expanding the derivative gives
\[
\frac{\mathrm{d}}{\mathrm{d}\lambda} H_{(\sigma,\mathbf{T})}\bigl(\tilde{P}_X^{[\lambda]}\bigr)\Big|_{\lambda=\lambda^\ast} -  \lambda^\ast\,\frac{\mathrm{d}}{\mathrm{d}\lambda}\E_{\tilde{P}_X^{[\lambda]}}[L(X)]\Big|_{\lambda=\lambda^\ast} = 0.
\]
Noting that the derivative of $\log Z(\lambda)$ satisfies
\[
\frac{\mathrm{d}}{\mathrm{d}\lambda}\log Z(\lambda)\Big|_{\lambda=\lambda^\ast} = \frac{\rm{d}}{\dd \lambda}H_{(\sigma,\mathbf{T})}\bigl(\tilde{P}^{[\lambda]}_X\bigr)\Big|_{\lambda=\lambda^{\ast}}-\E_{\tilde{P}^{[\lambda^{\ast}]}_X}[L(X)]-\lambda^{\ast}\frac{\dd}{\dd \lambda}\E_{\tilde{P}^{[\lambda]}_X}[L(X)]\Big|_{\lambda=\lambda^{\ast}},
\]
we immediately obtain
\[
\frac{\mathrm{d}}{\mathrm{d}\lambda}\log Z(\lambda)\Big|_{\lambda=\lambda^\ast} = -\E_{\tilde{P}_X^{[\lambda^\ast]}}[L(X)] = -\mu,
\]
which completes the proof.
\end{proof}
\begin{remark}\rm Theorem \ref{thm: the optimal lambda} can be interpreted as follows:
When we differentiate the maximum value $\log Z(\lambda)$ with respect to the Lagrange multiplier $\lambda$, we do not need to account for how the optimal distribution $\tilde{P}^{[\lambda]}_X$ itself changes as $\lambda$ varies.  This is connected with the envelope theorem \cite{milgrom2002envelope}, which states that since the chosen distribution already maximizes the objective at each fixed $\lambda$, any infinitesimal change in $\lambda$ only affects the optimal value through its direct appearance in the objective function.  In other words, there is no first-order contribution from the variation of the optimizer itself.
\end{remark}

\section{Hierarchical Invariances and Parameter Flows}\label{sec:hierarchical_invariances}
In this section, we present three scenarios where hierarchical invariances emerge, resulting in a substantial simplification of the iterative procedures from the previous section and notable computational benefits. Drawing inspiration from the concept of renormalization flow in renormalization group theory, we introduce parameter flow formulations for each case.
\subsection{Quadratic Modular Loss Function}\label{Self-similarity for multivariate Gaussians}
For some positive integers $k$ and $d$, assume that $N:=kd$ and ${X}=(X_1,\dots,X_N)$ take values on $\X=\mathbb{R}^N$. Assume that the loss function is a definite quadratic function $L(X)=X^\top QX$, where $Q\succ 0$ is a positive definite matrix. The vector ${X}$ is partitioned into $d$ blocks, each of size $k$: \[
{X} = ({X}_{[1]}, {X}_{[2]}, \ldots, {X}_{[d]}),
\]
where each block is defined ${X}_{[i]} := (X_{(i-1)k+1}, X_{(i-1)k+2}, \ldots, X_{ik})$ for $i = 1, \ldots, d$. For all $1\leq i \leq d$, we define decimation $T_i:\mathbb{R}^{(d-i+1)k}\to \mathbb{R}^{(d-i)k}$ as
\[T_i\left({X}_{[1]}, \ldots, {X}_{[d-i]}, {X}_{[d-i+1]}\right) = \left({X}_{[1]}, \ldots, {X}_{[d-i]}\right). \] 
Given the specified loss function, the initial distribution $P^{(1)}_{X^{(1)}}$ in Algorithm \ref{Hierarchical Entropy Algorithm} is a multivariate Gaussian. Due to the closure properties of multivariate Gaussians under marginalization and renormalization, it follows that each $P^{(i)}_{X^{(i)}}$ remains a multivariate Gaussian. Moreover, since multivariate Gaussians are also closed under conditioning and disintegration, the final distribution $P^{\ast}_{X}$ in \eqref{eq: definition of P star} is guaranteed to be multivariate Gaussian as well. Consequently, in this case, Algorithm \ref{Hierarchical Entropy Algorithm} reduces to computing the mean vector and covariance matrix of the distributions $P^{(i)}_{X^{(i)}}$. 
However, for $d\gg 1$, marginalization requires computing the covariance matrix and inverting the large precision matrix $Q$, which can be computationally prohibitive. In the following, we present an example of hierarchical invariance and self-similarity that allows us to bypass this computational burden.

Consider the symmetric type of multivariate Gaussian distributions defined as follows. This type of distribution possesses \emph{modularity}, a universal property of complex systems \cite{metzler2005multiscale}. 
\begin{definition}\label{symmetric block Toeplitz matrix}
	Given square matrices $A$ and $B$ with the same size and any integer $i\geq 1$, we define the following block Toeplitz matrix
	\begin{equation}
		\nonumber
		\mathcal{T}_i[A,B]:= \begin{pmatrix}
		A & B^T & B^T &\cdots  &B^T\\
		B & A & B^T &\cdots & B^T\\
		B & B & A &\cdots & B^T\\
		\vdots & \vdots & \vdots & \ddots & \vdots  \\
		B & B & B & \cdots &A
	\end{pmatrix}_{i\times i \textrm{ blocks}}.
	\end{equation}
We call a multivariate Gaussian distribution that has a precision matrix of the form above a \emph{modular} multivariate Gaussian distribution. 
\end{definition}
Note that the Schur complement theorem implies that if $\mathcal{T}_i[A,B]\succ 0$, then $A\succ 0$.


Before giving the main result of this subsection, we need the following lemma.
\begin{lemma}\label{lem: self-sim}
Let $i\geq 2$ and $({X}_{[1]},\dots, {X}_{[i]})$ be a modular multivariate Gaussian distribution with the precision matrix $\mathcal{T}_{i}[A,B]$, where $A$ is a symmetric invertible matrix. Then, $ ({X}_{[1]},\dots, {X}_{[i-1]})$ is also a modular multivariate Gaussian distribution with the precision matrix $\T_{i-1}[A',B']$, where
\begin{equation}\label{renormalization equations}
		\begin{cases}
		A'= A-B^TA^{-1}B,\\
		B'= B-B^TA^{-1}B.
	\end{cases}
	\end{equation}
\end{lemma}
\begin{proof}
	Let $M:= \T_i[A,B]$. The covariance matrix of $({X}_{[1]},\dots, {X}_{[i]})$
	is
	\begin{equation*}
		M^{-1}=\begin{pmatrix}
		N_{11} & N_{12}\\
		N_{21} & N_{22}
	\end{pmatrix},
	\end{equation*}
	where we have denoted the upper left $(i-1) \times (i-1)$ blocks with $N_{11}$, and $N_{22}$ has the same size as $A$ and $B$. Due to the well-known marginalization property of multivariate Gaussian distributions, $({X}_{[1]},\dots, {X}_{[i-1]})$ is also a multivariate Gaussian distribution with covariance matrix $N_{11}$.
If  we assume that 
 \( M = \begin{pmatrix} M_{11} & M_{12} \\ M_{21} & M_{22} \end{pmatrix} \) is a symmetric \( n \times n \) matrix, where \( M_{11} \) and \( M_{22} \) are \( p \times p \) and \( q \times q \), respectively, with \( n = p + q \),  then from a classic result in linear algebra we have 
    \(
    N_{11}^{-1} = M_{11} - M_{12} M_{22}^{-1} M_{21}.
    \) Therefore,  the precision matrix of $(X_{[1]},\dots,X_{[i-1]})$ is 
	\begin{align}
		N_{11}^{-1}&=M_{11}-M_{12}M_{22}^{-1}M_{21}\nonumber\\
				   &=\T_{i-1}[A,B]- \begin{pmatrix}
		B^T\\
		\vdots \\
		B^T
	\end{pmatrix}A^{-1}
	\begin{pmatrix}
		B & \cdots & B
	\end{pmatrix}\nonumber\\
	&=\T_{i-1}[A',B'].\nonumber
	\end{align} 
\end{proof} 
Note that since $B^TA^{-1}B$ is symmetric, $A'$ remains symmetric as well.
\begin{remark}\rm 
If  $B$  is not symmetric, the covariance matrix of  $({X}_{[1]},\dots, {X}_{[i]})$  or  $({X}_{[1]},\dots, {X}_{[i-1]})$  does not necessarily retain the modular structure described in Definition \ref{symmetric block Toeplitz matrix}, in contrast to the precision matrices established above.
\end{remark}

\begin{theorem}[Parameter flow]
    Let the loss function be $L(X)=X^\top QX$, where the matrix $Q=\T_{d}[A, B]$ is modular, symmetric, and positive definite. For $i=2,\dots, d$, let the matrices $A_i$ and $B_i$ evolve according to the recursive relations: 
\begin{equation}\nonumber
	\begin{cases}
		A_{i}=A_{i-1}-B_{i-1}^TA_{i-1}^{-1}B_{i-1},\\
		B_{i}=B_{i-1}-B_{i-1}^TA_{i-1}^{-1}B_{i-1},
	\end{cases}
\end{equation}
where the initial conditions are given by
\begin{equation}\nonumber
        \begin{cases}
		A_{1}=A,\\
		B_{1}=B.
	\end{cases}
    \end{equation} Then, all distributions $P^{(i)}_{X^{(i)}}$ in Algorithm \ref{Hierarchical Entropy Algorithm} have probability density functions of the similar form
    \begin{equation*}
        p^{(i)}_{X^{(i)}}\propto \exp\left(-X^TQ_iX\right),
    \end{equation*}
    where $Q_i = \T_{d-i+1}\bigl[\frac{\lambda}{\bl_i}A_i,\frac{\lambda}{\bl_i}B_i\bigr]$ is modular. 
\end{theorem} 

This theorem shows that to determine the distributions $P^{(i)}_{X^{(i)}}$, it suffices to compute the sub-blocks $(A_i,B_i)$ recursively, starting from $(A_1,B_1)$. Notably, for large $d\gg1$, the size of these matrices, $k$, is much smaller than that of the full precision matrix $Q$, which is of size $kd$. Thus, exploiting this hierarchical invariance significantly reduces computational complexity.

\begin{corollary}
Let $\lambda^\ast$ be the Lagrange multiplier for which $\tilde{P}_X^{[\lambda^\ast]}$ solves the optimization problem \eqref{eq: linear scalarization problem}. We have
\begin{align}
    \lambda^{\ast}=\frac{1}{\mu}\sum_{i=0}^{d-1} \frac{k(d-i)\sigma_{i+1}}{2\bl_{i+1}}.
\end{align}
\end{corollary}
\begin{proof} We have
\begin{align*}
    Z_1&= \int_{\mathbb{R}^{kd}} \exp{\left(-\frac{\lambda }{\sigma_1}x^\top Q x\right)}\mathrm{d}x\\
       &=\left(\frac{\sigma_1}{\lambda}\right)^{kd/2} \frac{\pi^{kd/2}}{\sqrt{\det Q}}.
\end{align*}
Thus, $\frac{d}{d\lambda} \log Z_1 = -\frac{kd}{2\lambda}$.
For all $1\leq i\leq d-1$, the precision matrix of $U^{(i)}_{X^{(i+1)}}$ is 
\begin{equation}\label{eq: precision of U_i}
    Q_i':=\T_{d-i}\Bigl[\frac{\lambda}{\bl_i}A_{i+1},\frac{\lambda}{\bl_i}B_{i+1}\Bigr].
\end{equation}
Let
\[
f_i(x) = \frac{\sqrt{\det Q'_i}}{(2\pi)^{k(d-i)/2}} \exp\left(-\frac{1}{2} x^\top Q'_i x\right),
\]
denote the probability density of $U^{(i)}_{X^{(i+1)}}$.
We can compute
\begin{align*}
    Z_{i+1} &= \int_{\mathbb{R}^{k(d-i)}} f_i(x)^{\frac{\bl_i}{\bl_{i+1}}} \, \dd x 
\\
&= \left(\frac{\sqrt{\det Q'_i}}{(2\pi)^{k(d-i)/2}}\right)^{\frac{\bl_i}{\bl_{i+1}}} \int_{\mathbb{R}^{k(d-i)}} \exp\left(-\frac{\bl_i}{2\bl_{i+1}} x^\top Q'_i x\right) \dd x\\
        &= (2\pi)^{\frac{k(d-i)}{2}\bigl(1-\frac{\bl_i}{\bl_{i+1}}\bigr)} \Bigl(\frac{\bl_i}{\bl_{i+1}}\Bigr)^{-\frac{k(d-i)}{2}} (\det Q'_i)^{\frac{\frac{\bl_i}{\bl_{i+1}}-1}{2}}.
\end{align*}
Thus, based on \eqref{eq: precision of U_i}, we derive
\begin{align*}
    \frac{\dd }{\dd \lambda}\log Z_{i+1}=-\frac{k(d-i)\sigma_{i+1}}{2\bl_{i+1}\lambda}.
\end{align*}
Invoking Theorem \ref{thm: the optimal lambda}, we complete the proof.
\end{proof}
\begin{remark} \rm
    If the loss function includes a linear term, i.e.,  
\begin{equation*}
    L(X) = X^T Q X + g^T X,
\end{equation*}  
then, we can rewrite the loss function as  
\begin{equation*}
    L(X) = (X - \mu)^T Q (X - \mu) + C,
\end{equation*}  
where \(C\) is a constant independent of \(X\), and mean \(\mu\) can be efficiently computed using the Block-Levinson recursion \cite{musicus1988levinson} by solving  
\begin{equation*}
    -2Q\mu = g.
\end{equation*}  
With this reformulation, we can directly apply the theorem stated above to shifted versions of the distributions by replacing $X-\mu$ with $X'$.

\end{remark}
\begin{remark} \rm 
    The Hessian of residual neural networks, evaluated at the origin with respect to the parameters, exhibits the block Toeplitz structure $\mathcal{T}_d[A,B]$
  \cite{li2016demystifying}. Consequently, when the empirical loss function of the neural network is approximated by its second-order Taylor expansion--equivalently, when the generalized posterior is approximated as a multivariate Gaussian--it possesses the invariance property established in the previous theorem.
\end{remark}

\subsection{Logarithmic Loss Function}
For a positive integer $d$, assume that $N=2^d$ and $X=(X_1,\dots,X_N)$. Let the space be 
\[\mathcal{X}:=\{{x}=(x_1,\dots,x_N): x_1+\dots+x_N=1, 0\leq x_i\leq 1\}.\] For all $1\leq i \leq d$, we define the hierarchical transformation $T_i:\mathbb{R}^{\frac{N}{2^{i-1}}}\to \mathbb{R}^{\frac{N}{2^i}}$ as
\begin{align*}
    T_i\left(x_1,\dots,x_{\frac{N}{2^{i-1}}}\right):=\left(x_1+x_2,x_3+x_4,\dots,x_{\frac{N}{2^{i-1}}-1}+x_{\frac{N}{2^{i-1}}}\right).
\end{align*}
Let $\log x^{(i)}$ denote the component-wise logarithm of vector $x^{(i)}$. Assume that $\boldsymbol{\alpha}=(\alpha_1,\dots,\alpha_N)$.  
The following result shows hierarchical invariance in this setting:
\begin{theorem}[Parameter flow] Suppose the loss function has the logarithmic form:
\(
L({x}) = -\boldsymbol{\alpha}^\top \log x
\), where $\boldsymbol{\alpha}$ is a vector with all non-negative entries. Then, all distributions $P^{(i)}_{X^{(i)}}$ in Algorithm \ref{Hierarchical Entropy Algorithm} have probability density functions of the similar form
    \begin{equation*}
        p^{(i)}_{X^{(i)}}\propto \exp\left(-\boldsymbol{\alpha_{(i)}}^\top\log x^{(i)}\right).
    \end{equation*}
Moreover, for $i=2,\dots, d$, the vectors $\boldsymbol{\alpha}_{(i)}$ evolve according to the recursive relations: 
\begin{equation}\nonumber
		\boldsymbol{\alpha}_{(i)}=\frac{\bl_{i-1}}{\bl_i}T_{i-1}\left(\boldsymbol{\alpha}_{(i-1)}+\boldsymbol{1}_{(i-1)} \right)-\boldsymbol{1}_{(i)},
\end{equation}
where $\boldsymbol{1}_{(i)}$ denotes the all-one vector with the same length as $\boldsymbol{\alpha}_{(i)}$, and the initial condition is given by
\begin{equation}\nonumber
        \boldsymbol{\alpha}_{(1)}=\frac{\lambda}{\sigma_1}\boldsymbol{\alpha}.
    \end{equation}
\end{theorem}
\begin{proof}
    We begin by noting that 
\begin{align*}
    p_{X^{(1)}}^{(1)}({x})&\propto \exp\left(-\frac{\lambda }{\sigma_1}L({x})\right)\\
                        &= \exp\left(\frac{\lambda \sum_{i=1}^N \alpha_i \log x_i}{\sigma_1}\right)\\
                        &=\prod_{i=1}^N x_i^{\frac{\lambda\alpha_i}{\sigma_1}}.
\end{align*}
This shows that $p_{X^{(1)}}^{(1)}({x})$ follows a Dirichlet distribution with concentration parameters \(\boldsymbol{\alpha}_{(1)}+\boldsymbol{1}_{(1)}.\) 
To complete the result, we apply the aggregation property of Dirichlet distributions, which states that if $\boldsymbol{\alpha}$ is a vector with all entries larger than or equal to $-1$ and 
    if  
\[
(X_1, X_2, \dots, X_N) \sim \text{Dir}(\alpha_1, \alpha_2, \dots, \alpha_N),
\]
then, replacing two adjacent components \( X_i \) and \( X_{i+1} \) by their sum  
results in another Dirichlet distribution as
\[
(X_1, \dots, X_{i}+X_{i+1}, \dots, X_N) \sim \text{Dir}(\alpha_1, \dots, \alpha_i + \alpha_{i+1}, \dots, \alpha_N).
\]
 Using this property,  
we obtain 
\[U^{(1)}_{X^{(2)}}\sim \text{Dir}\bigl(T_1\bigl(\boldsymbol{\alpha}_{(1)}+\boldsymbol{1}_{(1)}\bigr)\bigr).\] 
Based on this result, we deduce that \[P^{(2)}_{X^{(2)}}\sim \text{Dir}\Bigl(\frac{\bar{\sigma}_1}{\bar{\sigma}_2}T_1\bigl(\boldsymbol{\alpha}_{(1)}+\boldsymbol{1}_{(1)}\bigr)\Bigr).\] 
The proof is completed by applying this argument inductively.
\end{proof}
\subsection{Nearest-Neighbors Loss Function}\label{subsec: one-dimensional Ising}
For a positive integer $d$, assume that $N=2^d$ and $X=(X_1,\dots,X_N)$. Let the space be 
\(\mathcal{X}:=\{-1,+1\}^N.\) For all $1\leq i \leq d$, we define the hierarchical transformation $T_i:\mathbb{R}^{\frac{N}{2^{i-1}}}\to \mathbb{R}^{\frac{N}{2^i}}$ as
\begin{align*}
    T_i\left(x_1,\dots,x_{\frac{N}{2^{i-1}}}\right):=\left(x_1,x_3,\dots,x_{\frac{N}{2^{i-1}}-1}\right).
\end{align*}
Thus, the hierarchical transformations are defined as decimating the even-numbered spins.

For any vector $y$ of size $\ell$, we define the cyclic sum
\[\sum_{\text{cyc}}y_jy_{j+1} := \sum_{j=1}^{\ell}y_jy_{j+1}, \]
where $y_{\ell+1}:=y_1$.
The following result demonstrates hierarchical invariance in this setting:
\begin{theorem}[Parameter flow]\label{thm: parameter flow ising}
Suppose the loss function takes the form of nearest-neighbor interactions, as commonly assumed in Ising models:
\begin{equation}\nonumber
	L(x)= -J\sum_{\text{cyc}}x_jx_{j+1},
\end{equation}
where $J>0$ is called the interaction strength and we define $X_{N+1}:= X_1$, called the periodic boundary condition. Then, all distributions $P^{(i)}_{X^{(i)}}$ in Algorithm \ref{Hierarchical Entropy Algorithm} have probability distributions of the similar nearest-neighbor form
    \begin{equation*}
        P^{(i)}_{X^{(i)}}\left(x^{(i)}\right)\propto \exp\left(\theta_i \sum_{\text{cyc}}x^{(i)}_jx^{(i)}_{j+1}\right).
    \end{equation*}
Moreover, for $i=2,\dots, d$, the parameters $\theta_i$ evolve according to the recursive relations: 
\begin{equation}\label{renormalization flow 3}
		\theta_i= \frac{\bl_{i-1}}{2\bl_i}\log \cosh \left(2\theta_{i-1}\right),
\end{equation}
where the initial condition is given by
\begin{equation}\nonumber
        \theta_1=\frac{\lambda J}{\sigma_1}.
    \end{equation}
\end{theorem}
\begin{proof}
We have \[P^{(1)}_{X^{(1)}}(x)\propto \exp\left(-\frac{\lambda}{\sigma_1}L({x})\right)= \exp\left(\theta_1\sum_{\text{cyc}}x_jx_{j+1}\right).\] 
We can derive
\begin{align*}
	U^{(1)}_{X^{(2)}}\left(x^{(2)}\right) &\propto \sum_{x_2,x_4,\dots,x_N}\exp\left( \theta_1 \sum_{j=1}^{N}x_jx_{j+1}\right)\\
						 & =\sum_{x_2,x_4,\dots,x_N}\prod_{j=1}^{N/2}\exp\left(\theta_1\left(x_{2j-1}x_{2j}+x_{2j}x_{2j+1} \right)\right)\\
						 &= \prod_{j=1}^{N/2}\left(\sum_{x_{2j}} \exp\left(\theta_1\left(x_{2j-1}x_{2j}+x_{2j}x_{2j+1} \right)\right)\right).
\end{align*}
We use the fact that for all possible values of $s_1,s_2\in\{-1,+1\}$ and any $\theta>0$, (see \cite{maris1978teaching}) 
    \[\exp(\theta(s_1+s_2))+\exp(-\theta(s_1+s_2))=\left(2\cosh^{1/2}(2\theta)\right)\exp\left(\frac{1}{2}\log \cosh (2\theta)s_1s_2\right).\]
Applying this result, for any $j=1,\dots,N/2$, we obtain
\begin{align*}
    \sum_{x_{2j}} \exp\left(\theta_1\left(x_{2j-1}x_{2j}+x_{2j}x_{2j+1} \right)\right)
    &= \left(2\cosh^{1/2}(2\theta)\right)\exp\left(\theta'x_{2j-1}x_{2j+1}\right),
\end{align*}
where $\theta' = \frac{1}{2}\log \cosh (2\theta)$. 
Thus, we have
\begin{align*}
	U^{(1)}_{X^{(2)}}\left(x^{(2)}\right) &\propto \prod_{j=1}^{N/2}\left(2\cosh^{1/2}(2\theta)\right)\exp\left(\theta'x_{2j-1}x_{2j+1}\right)\\
	 & = \left(2\cosh^{1/2}(2\theta)\right)^{\frac{N}{2}}\prod_{j=1}^{N/2}\exp\left(\theta'x_{2j-1}x_{2j+1}\right)\\
     &=\left(2\cosh^{1/2}(2\theta)\right)^{\frac{N}{2}}\exp\left(\theta'\sum_{\text{cyc}}x^{(2)}_{j}x^{(2)}_{j+1}\right)\\
     &\propto \exp\left(\theta'\sum_{\text{cyc}}x^{(2)}_{j}x^{(2)}_{j+1}\right).
\end{align*}
 Thus, 
 \begin{align*}
     P^{(2)}_{X^{(2)}}\left(x^{(2)}\right)&\propto \exp\left(\frac{\bl_1\theta'}{\bl_{2}}\sum_{\text{cyc}}x^{(2)}_{j}x^{(2)}_{j+1}\right)\\
     &=\exp\left(\theta_2\sum_{\text{cyc}}x^{(2)}_{j}x^{(2)}_{j+1}\right).
 \end{align*}
 Using this argument inductively completes the proof.
\end{proof}
    This hierarchical invariance provides a significant computational advantage even in sampling from the classical maximum entropy distribution, where $\bl_d=\dots=\bl_1>0$, a non-trivial observation. To efficiently sample from the distribution $\tilde{P}^{[\lambda]}_X$ defined over the vast space $\mathcal{X}$, standard methods like Markov Chain Monte Carlo (MCMC) can be computationally expensive, often requiring the rejection of numerous initial samples to ensure proper mixing. Instead, by leveraging the self-similarity of the distribution, we propose an alternative approach that enables direct sampling without discarding any samples, significantly improving efficiency.
    
    For fixed values of $X^{(2)}=x^{(2)}=(x_1,x_3,\dots,x_{N-1})$, the conditional distribution\\ $P^{(1)}_{X^{(1)}|X^{(2)}=x^{(2)}}$ takes the form 
\begin{align*}
	P^{(1)}_{X^{(1)}|X^{(2)}=x^{(2)}} & \propto \exp\left(\theta_1\sum_{j=1}^{N/2} X_{2j}(x_{2j-1}+x_{2j+1})\right),
\end{align*}
which factorizes as
\[P^{(1)}_{X^{(1)}|X^{(2)}=x^{(2)}}=\prod_{j=1}^{N/2} \exp(\theta_1 X_{2j}(x_{2j-1}+x_{2j+1})).\]
Since this expression factorizes over the variables $X^{(1)}/X^{(2)}=\{X_2,X_4,\dots,X_N\}$, each spin $X_2$, $X_4,$ $\dots,$ $X_N$ can be sampled independently as Bernoulli random variables. This allows for efficient sampling using the inverse transformation method, avoiding the need to discard samples. By using Theorem \ref{thm: parameter flow ising} and applying this argument inductively, we conclude that sampling from $\tilde{P}^{[\lambda]}_X$ can be performed efficiently without rejection.

\section{Hierarchical Minimum Relative Entropy}\label{sec: relative entropy}
In this section, given $\X$, $\mathbf{T}$, and an arbitrary distribution $Q_X$, we study the Pareto optimal solutions of the related \emph{hierarchical minimum relative entropy} problem:
\begin{equation}\label{eq: relative entropy vector opt problem}
    \argmin_{P_X:~ \E[L(X)]=\mu} \bigl(D\bigl(P_{X^{(1)}}\big\|Q_{X^{(1)}}\bigr),D\bigl(P_{X^{(2)}}\big\|Q_{X^{(2)}}\bigr),\dots ,D\bigl(P_{X^{(d)}}\big\|Q_{X^{(d)}}\bigr)\bigr),
\end{equation}
where $P_{X^{(i)}}$ and $Q_{X^{(i)}}$ are the pushforward measures when $P_X$ and $Q_X$ are passed through the sequence of transformations $\mathbf{T}:=(T_i)_{i=1}^{d-1}$, respectively. Based on an argument similar to Section \ref{sec:problem_description}, we analyze the associated hierarchical relative entropy regularization problem
\begin{equation}\label{eq: hierarchical relative entropy min}
    \argmin_{P_X} \left\{D_{(\sigma,\mathbf{T})}(P_X\|Q_X)+\lambda\E[L(X)]\right\},
\end{equation}
and show how the solution can be obtained by a generalized version of the renormalization group. This formulation represents the Lagrangian for minimizing $D_{(\sigma, \mathbf{T})}(P_X\|Q_X)$ subject to the mean constraint \(\E[L(X)] = \mu\).

Consider the following generalization of the renormalization operator:
\begin{definition}[Generalized renormalization] Let $0\leq \theta \leq 1$ and $P^{(1)}_X$ and $P^{(2)}_X$ be continuous distributions on a set $\mathcal{A}$ with probability density functions $p_1(x)$ and $p_2(x)$. We define the generalized renormalization operator $(Q_X,\tilde{Z})=\mathcal{G}\bigl(P^{(1)}_X,P_X^{(2)};\theta\bigr)$ where
\begin{equation*}
    \tilde{Z}:=\int_{\mathcal{A}}(p_1(x))^{\theta}(p_2(x))^{1-\theta}\dd x,
\end{equation*}
and $Q_X$ is a distribution with the probability density function  
\begin{equation}
	\nonumber
	q_X(x):= \frac{(p_1(x))^{\theta}(p_2(x))^{1-\theta}}{\tilde{Z}}, \textrm{ for all }x\in \mathcal{A}.
\end{equation}
The generalized renormalization operator is defined analogously for discrete random variables except by replacing the integrals with sums.
\end{definition}
The distribution $Q_X$ defined above is known as the \emph{generalized escort distribution} in the statistical physics literature; see, e.g., \cite{bercher2012simple}.

For the proof of the following lemma, see \cite[Theorem 30]{van2014renyi}.
\begin{lemma}\label{lem: Tilted Renyi lemma main}
Let $\theta\in [0,1]$. For any distributions $P, Q_1$ and $Q_2$, let $(\tilde{Q},\tilde{Z})=\mathcal{G}(Q_1,Q_2;\theta)$ be the output of the generalized renormalization operator. Then,
\begin{align}
	\theta D(P\|Q_1)+(1-\theta)D(P\|Q_2)= D(P\|\tilde{Q})-\log \tilde{Z}.\nonumber
\end{align}	
\end{lemma}
The following result is the counterpart to the Gibbs variational principle for relative entropy:
\begin{lemma}\label{lem: Gibbs relative entropy}
  Let ${\mathcal A}$ be a continuous set and function $f: {\mathcal A}\rightarrow \mathbb{R}$ be such that 
  \begin{align*}
  	\tilde{Z}:=\int_{x\in {\mathcal A}}\exp\left(-\lambda f(x)\right)q_X(x)\mathrm{d}x<\infty.
  \end{align*}
  For any distribution $P_{X}$ defined on ${\mathcal A}$ with probability density function $p_X(x)$, such that $X\sim P_X$, we have
  \begin{equation}
      \nonumber
     D(P_X\|Q_X) + \lambda\mathbb{E}[f(X)]= D\left(P_{X}\middle\|\tilde{P}_{X}\right)-\log \tilde{Z} ,
  \end{equation}
  where $\tilde{P}_X$ is the distribution with probability density function
  \[\tilde{p}_{X}(x):=\frac{\exp\left(-\lambda f(x)\right)q_{X}(x)}{\tilde{Z}}, \textrm{ for all } x\in {\mathcal A}.\]
\end{lemma}
\begin{proof}
We can write 
\begin{align*}
    D(P_X\|Q_X)+\lambda \mathbb{E}[f(X)]&= \int_{\mathcal{A}} p_X(x)\log\left(\frac{p_X(x)}{q_X(x)}\right)\mathrm{d}x+ \lambda\int_{\mathcal{A}} f(x)p_X(x)\mathrm{d}x \\
                                        &= \int_{\mathcal{A}} p_X(x)\log\left(\frac{p_X(x)}{\exp\left(-\lambda{f(x)}\right)q_X(x)}\right)\mathrm{d}x\\
                                        &= D\left(P_{X}\middle\|\tilde{P}_{X}\right) - \log \tilde{Z}.
\end{align*}
\end{proof}

In \cite{asadi2020chaining}, it was shown that when $\mathbf{T}$ is a sequence of decimation transformations, the optimization problem (\ref{eq: hierarchical relative entropy min}) has a unique minimizer, which can be efficiently computed using the Marginalize-Tilt (MT) algorithm. Here, we extend this result to a much broader class of hierarchical transformations, proving that the solution remains unique and can be obtained using Algorithm \ref{alg: generalized renormalization}, a generalization of the MT algorithm, and we derive a condition on the temperature for which the mean constraint $\E[L(X)]=\mu$ is satisfied.

Again, we assume that all alphabets are standard Borel spaces, ensuring the existence of regular conditional probabilities and well-defined reverse random transformations. 

\begin{algorithm}[t]
    \caption{Generalized Renormalization Group}\label{alg: generalized renormalization}
\begin{algorithmic}[1]
    \State \textbf{let} \(Z_1:=\int \exp{\left(-\frac{\lambda }{\sigma_1}L(x)\right)}q_{X}(x)\mathrm{d}x \) \\ and \(p^{(1)}_{X}(x):=\frac{\exp{\left(-\frac{\lambda }{\sigma_1}L(x)\right)}q_{X}(x)}{Z_1}\)
    \Comment{Initialization}
    \For{\(i = 1 \text{ to } d-1\)}
        \State \(U^{(i)}_{X^{(i+1)}} \gets (T_i)_{\#}P^{(i)}_{X^{(i)}}\) 
        \Comment{Pushforward (coarse-graining)}
        \State \(\left(P^{(i+1)}_{X^{(i+1)}},Z_{i+1}\right) \gets \mathcal{G}\left(U^{(i)}_{X^{(i+1)}},Q_{X^{(i+1)}};\frac{\bar{\sigma}_{i}}{\bar{\sigma}_{i+1}}\right)\) 
        \Comment{Generalized renormalization}
    \EndFor
\end{algorithmic}
\end{algorithm}
\begin{theorem}\label{hierarchical relative entropy RG theorem} Consider Algorithm \ref{alg: generalized renormalization} and the definition of $\tilde{P}^{[\lambda]}_X$ and $Z(\lambda)$ identical to equations \eqref{eq: definition of P star} and \eqref{eq: definition of Z}. 
For any $P_X$ defined on $\mathcal{X}$ where $X\sim P_X$, we have
\begin{equation*}
		D_{(\sigma,\mathbf{T})}(P_X\|Q_X)+\lambda\E[L(X)]=D_{(\sigma,\mathbf{T})}\Bigl(P_X\Big\|\tilde{P}^{[\lambda]}_X\Bigr)-\log Z(\lambda).
	\end{equation*}
\end{theorem}
\begin{proof}
The proof follows a similar structure to that of Theorem \ref{Max hierarchical Shannon entropy}, but instead of using Lemmas \ref{Gibbs is the maximizer entropy} and \ref{Shannon entropy KL sum}, we apply Lemmas \ref{lem: Tilted Renyi lemma main} and \ref{lem: Gibbs relative entropy}. Similarly, Theorem \ref{hierarchical relative entropy RG theorem} can be viewed as a hierarchical extension of Lemma \ref{lem: Gibbs relative entropy}.
\end{proof}
\begin{corollary}
    The distribution $\tilde{P}^{[\lambda]}_X$ is the unique solution to the minimization problem (\ref{eq: hierarchical relative entropy min}) and 
\begin{equation*}
    \min_{P_X} \left\{D_{(\sigma,\mathbf{T})}(P_X\|Q_X)+\lambda\E[L(X)]\right\}= -\log Z(\lambda).
\end{equation*}
\end{corollary}
We can similarly solve for $\lambda$ from the constraint given in the following theorem:
\begin{theorem}
Let $\lambda^\ast$ be the Lagrange multiplier for which $\tilde{P}_X^{[\lambda^\ast]}$ solves the optimization problem \eqref{eq: relative entropy vector opt problem}. Then $\lambda^\ast$ is determined by the condition
\[
\frac{\mathrm{d}}{\mathrm{d}\lambda}\log Z(\lambda)\Big|_{\lambda=\lambda^\ast} = \mu.
\]
\end{theorem}

\section{Conclusion}\label{sec:conclusions}
This paper introduced the problem of hierarchical maximum entropy, extending the classical maximum entropy principle to systems with hierarchical structures. By addressing the interplay between uncertainty and constraints across multiple levels, this framework broadens the scope of maximum entropy methods to better capture structures that arise naturally in many real-world systems. We established a theoretical foundation for hierarchical maximum entropy and demonstrated that its solutions can be obtained through the renormalization group and disintegration operations. We also provided examples with hierarchical invariances, which significantly simplified the iterative renormalization group procedures by introducing the concept of parameter flows for quadratic modular loss functions, logarithmic loss functions, and nearest-neighbor loss functions. These examples illustrate some applications of the hierarchical maximum entropy framework. The insights from this work connect key ideas from probability theory, information theory, and statistical mechanics, offering a perspective for studying systems with hierarchical structures. Beyond theoretical contributions, this framework has the potential for applications in fields such as machine learning, statistical inference, and the analysis of physical and biological systems. Future research could explore additional classes of loss functions and alphabets, identify further instances of hierarchical invariances and computational efficiencies, investigate computational techniques for high-dimensional implementations, and extend the framework to incorporate alternative entropy measures. These advancements could enhance both the theoretical and practical impact of the hierarchical maximum entropy framework.

\section*{Acknowledgments}
I would like to thank Emmanuel Abbe from EPFL for valuable discussions, which formed the basis of this work, and Siddhartha Sarkar from the Max Planck Institute for the Physics of Complex Systems for insightful input on renormalization group theory. This research is supported by Leverhulme Trust grant ECF-2023-189 and
Isaac Newton Trust grant 23.08(b).

\bibliographystyle{plain} 
\bibliography{My_Biblio.bib}       

\begin{thebibliography}{10}

\bibitem{alquier2024user}
Pierre Alquier.
\newblock User-friendly introduction to {PAC}-{B}ayes bounds.
\newblock {\em Foundations and Trends{\textregistered} in Machine Learning},
  17(2):174--303, 2024.

\bibitem{asadi2020chaining}
Amir~R. Asadi and Emmanuel Abbe.
\newblock Chaining meets chain rule: Multilevel entropic regularization and
  training of neural networks.
\newblock {\em Journal of Machine Learning Research}, 21(139):1--32, 2020.

\bibitem{bercher2012simple}
Jean-Francois Bercher.
\newblock A simple probabilistic construction yielding generalized entropies
  and divergences, escort distributions and q-gaussians.
\newblock {\em Physica A: Statistical Mechanics and its Applications},
  391(19):4460--4469, 2012.

\bibitem{berger1996maximum}
Adam~L Berger, Vincent J~Della Pietra, and Stephen A~Della Pietra.
\newblock A maximum entropy approach to natural language processing.
\newblock {\em Computational Linguistics}, 22(1):39--71, 1996.

\bibitem{Catoni2007}
Olivier Catoni.
\newblock {PAC-Bayesian} supervised classification: The thermodynamics of
  statistical learning.
\newblock {\em Lecture Notes-Monograph Series}, 56:i--163, 2007.

\bibitem{csiszar1975divergence}
Imre Csisz{\'a}r.
\newblock I-divergence geometry of probability distributions and minimization
  problems.
\newblock {\em The Annals of Probability}, pages 146--158, 1975.

\bibitem{de2021diffusion}
Valentin De~Bortoli, James Thornton, Jeremy Heng, and Arnaud Doucet.
\newblock Diffusion {S}chr{\"o}dinger bridge with applications to score-based
  generative modeling.
\newblock {\em Advances in Neural Information Processing Systems},
  34:17695--17709, 2021.

\bibitem{weinan2011principles}
W.~E.
\newblock {\em Principles of Multiscale Modeling}.
\newblock Cambridge University Press, 2011.

\bibitem{faden1985existence}
Arnold~M. Faden.
\newblock The existence of regular conditional probabilities: necessary and
  sufficient conditions.
\newblock {\em The Annals of Probability}, pages 288--298, 1985.

\bibitem{fogedby1992phase}
Hans~C Fogedby.
\newblock On the phase space approach to complexity.
\newblock {\em Journal of Statistical Physics}, 69(1-2):411--425, 1992.

\bibitem{geoffrion1968proper}
Arthur~M Geoffrion.
\newblock Proper efficiency and the theory of vector maximization.
\newblock {\em Journal of Mathematical Analysis and Applications},
  22(3):618--630, 1968.

\bibitem{grunwald2004game}
Peter~D. Gr{\"u}nwald and A.~Philip Dawid.
\newblock Game theory, maximum entropy, minimum discrepancy and robust bayesian
  decision theory.
\newblock {\em Annals of Statistics}, 32(4):1367--1433, 2004.

\bibitem{hwang2012multiple}
C-L Hwang and Abu Syed~Md Masud.
\newblock {\em Multiple objective decision making -- methods and applications:
  a state-of-the-art survey}, volume 164.
\newblock Springer Science \& Business Media, 2012.

\bibitem{jaynes1957information}
Edwin~T Jaynes.
\newblock Information theory and statistical mechanics.
\newblock {\em Physical Review}, 106(4):620, 1957.

\bibitem{jaynes1957information2}
Edwin~T Jaynes.
\newblock Information theory and statistical mechanics. ii.
\newblock {\em Physical Review}, 108(2):171, 1957.

\bibitem{jona2001renormalization}
Giovanni Jona-Lasinio.
\newblock Renormalization group and probability theory.
\newblock {\em Physics Reports}, 352(4-6):439--458, 2001.

\bibitem{kadanoff1975numerical}
Leo~P Kadanoff and Anthony Houghton.
\newblock Numerical evaluations of the critical properties of the
  two-dimensional ising model.
\newblock {\em Physical Review B}, 11(1):377, 1975.

\bibitem{koralov2007theory}
Leonid Koralov and Yakov~G Sinai.
\newblock {\em Theory of Probability and Random Processes}.
\newblock Springer Science \& Business Media, 2007.

\bibitem{li2016demystifying}
Sihan Li, Jiantao Jiao, Yanjun Han, and Tsachy Weissman.
\newblock Demystifying resnet.
\newblock {\em arXiv preprint arXiv:1611.01186}, 2016.

\bibitem{maris1978teaching}
Humphrey~J Maris and Leo~P Kadanoff.
\newblock Teaching the renormalization group.
\newblock {\em American Journal of Physics}, 46(6):652--657, 1978.

\bibitem{metzler2005multiscale}
Richard Metzler and Yaneer Bar-Yam.
\newblock Multiscale complexity of correlated gaussians.
\newblock {\em Physical Review E}, 71(4):046114, 2005.

\bibitem{miettinen1999nonlinear}
Kaisa Miettinen.
\newblock {\em Nonlinear multiobjective optimization}, volume~12.
\newblock Springer Science \& Business Media, 1999.

\bibitem{migdal1975recursion}
Alexander~A Migdal.
\newblock Recursion equations in gauge field theories.
\newblock {\em Sov. Phys. JETP}, 42(3):413--418, 1975.

\bibitem{milgrom2002envelope}
Paul Milgrom and Ilya Segal.
\newblock Envelope theorems for arbitrary choice sets.
\newblock {\em Econometrica}, 70(2):583--601, 2002.

\bibitem{musicus1988levinson}
Bruce~Ronald Musicus.
\newblock Levinson and fast choleski algorithms for toeplitz and almost
  toeplitz matrices.
\newblock {\em Research Laboratory of Electronics Technical Report,
  Massachusetts Institute of Technology
  \url{https://dspace.mit.edu/handle/1721.1/4954}}, 1988.

\bibitem{nigam1999using}
Kamal Nigam, John Lafferty, and Andrew McCallum.
\newblock Using maximum entropy for text classification.
\newblock In {\em IJCAI-99 workshop on machine learning for information
  filtering}, volume~1, pages 61--67. Stockholom, Sweden, 1999.

\bibitem{nutz2022entropic}
Marcel Nutz and Johannes Wiesel.
\newblock Entropic optimal transport: Convergence of potentials.
\newblock {\em Probability Theory and Related Fields}, 184(1):401--424, 2022.

\bibitem{phillips2004maximum}
Steven~J Phillips, Miroslav Dud{\'\i}k, and Robert~E Schapire.
\newblock A maximum entropy approach to species distribution modeling.
\newblock In {\em Proceedings of the twenty-first International Conference on
  Machine Learning}, page~83, 2004.

\bibitem{ratnaparkhi-1996-maximum}
Adwait Ratnaparkhi.
\newblock A maximum entropy model for part-of-speech tagging.
\newblock In {\em Conference on Empirical Methods in Natural Language
  Processing}, 1996.

\bibitem{soofi2000principal}
Ehsan~S Soofi.
\newblock Principal information theoretic approaches.
\newblock {\em Journal of the American Statistical Association},
  95(452):1349--1353, 2000.

\bibitem{van2014renyi}
Tim Van~Erven and Peter Harremos.
\newblock R{\'e}nyi divergence and {Kullback-Leibler} divergence.
\newblock {\em IEEE Transactions on Information Theory}, 60(7):3797--3820,
  2014.

\bibitem{wilson1983renormalization}
Kenneth~G Wilson.
\newblock The renormalization group and critical phenomena.
\newblock {\em Reviews of Modern Physics}, 55(3):583, 1983.

\bibitem{wilson1974renormalization}
Kenneth~G Wilson and John Kogut.
\newblock The renormalization group and the $\epsilon$ expansion.
\newblock {\em Physics reports}, 12(2):75--199, 1974.

\bibitem{zhang1991complexity}
Yi-Cheng Zhang.
\newblock Complexity and $1/f$ noise: A phase space approach.
\newblock {\em Journal de Physique I}, 1(7):971--977, 1991.

\end{thebibliography}

\end{document}